\documentclass[%
 reprint,
nofootinbib,
 amsmath,amssymb,
 aps,
]{revtex4-2}
\pdfoutput=1
\usepackage{physics}
\usepackage{tikz}
\usepackage{xcolor}

\usepackage[T1]{fontenc}
\usepackage{dsfont,hyperref}
\usepackage{amsthm,amsfonts,pifont}

\newtheorem{theorem}{Theorem}

\newtheorem{definition}{Definition}
\newtheorem{example}{Example}
\newtheorem{lemma}{Lemma}
\newtheorem{remark}{Remark}

\usepackage{thm-restate}

\begin{document}

\title{Connecting XOR and XOR* games}

\author{
Lorenzo Catani$^{1}$}
\email{lorenzo.catani@tu-berlin.de}
\author{Ricardo Faleiro$^{2}$} \email{ricardofaleiro@tecnico.ulisboa.pt}
\author{Pierre-Emmanuel Emeriau$^{3}$}
\author{Shane Mansfield$^3$}
\author{Anna Pappa$^1$}

\affiliation{$^{1}$Electrical Engineering and Computer Science Department,
Technische Universitat Berlin, 10587 Berlin, Germany\\
$^{2}$Instituto de Telecomunicações, Avenida Rovisco Pais 1,049-001, Lisboa, Portugal\\
$^{3}$Quandela,
 7 Rue Léonard de Vinci,
 91300 Massy, France\\}

\begin{abstract}
 In this work we focus on two classes of games: XOR nonlocal games and XOR* sequential games with monopartite resources. XOR games have been widely studied in the literature of nonlocal games, and we introduce XOR* games as their natural counterpart within the class of games where a resource system is subjected to a sequence of controlled operations and a final measurement.  Examples of XOR* games are $2\rightarrow 1$ quantum random access codes (QRAC) and the CHSH* game introduced by Henaut et al.~in [PRA \textbf{98},060302(2018)].
We prove, using the diagrammatic language of process theories, that under certain assumptions these two classes of games can be related via an explicit theorem that connects their optimal strategies, and so their classical (Bell) and quantum (Tsirelson) bounds. We also show that two of such assumptions -- the reversibility of transformations and the bi-dimensionality of the resource system in the XOR* games -- are strictly necessary for the theorem to hold by providing explicit counterexamples.
We conclude with several examples of pairs of XOR/XOR* games and by discussing in detail the possible resources that power the quantum computational advantages in XOR* games. 
\end{abstract}
\maketitle

\section{Introduction}

Computational tasks for which quantum strategies outperform classical ones have long been and are still an important focus of study.
A well-known example is the CHSH game, a way of recasting the Clauser-Horne-Shimony-Holt (CHSH) formulation \cite{CHSH69} of Bell's celebrated theorem \cite{Bell64} into a game for which quantum strategies can provide an advantage over classical ones. It is a cooperative game with two players, Alice and Bob, who agree on a strategy before the game begins but are separated and unable to communicate with each other once the game is in progress.
A referee poses one of two binary questions to each player, and the players win if the sum of their answers equals the product of the questions (arithmetic is modulo $2$). 
The bound on the performances of classical strategies is known as the Bell bound, while in the quantum case the maximum winning probability is known as the Tsirelson bound \cite{Cirelson1980}.
The CHSH game is of great importance due to the dependence of its optimal success probability on the underlying theory describing the implemented strategies, which gives us a tool to experimentally distinguish between distinct physical theories.
Scenarios that involve violations of the Bell bound are said to manifest ``nonlocality'' (meaning that the statistics they show are inconsistent with a description in terms of a local hidden variable model \cite{Bell64}), and  the Tsirelson bound can be interpreted as a quantifier of how nonlocal quantum theory is within the broader landscape of all possible no-signaling theories.
Nonlocality is also known to be a useful resource for applications in quantum technology, \textit{e.g.},\ in device-independent cryptography \cite{MY04,P10, VV14} and it can be considered to be a special form of the more general notion of contextuality  \cite{Fine,abramskysheaf},
known to be a necessary ingredient for quantum advantage and speedups in a variety of settings \cite{Spekkens2009,Abbott2012,Raussendorf2013,Zhang2013,Howard2014,Guhne2014,Delfosse2015,Chailloux2016,Raussendorf2017,Vega2017,Tavakoli2017,Oestereich2017,CataniBrowne2018,Mansfield2018, Schmid2018,Frembs2018,Raussendorf2019,Saha2019,SahaAnubhav2019,Bharti2019,LostaglioSenno2020,lostaglio2020certifying,Yadavalli2020,Um2020,Emeriau2020,Flatt2021,Roch2021}.
 
Inspired by the CHSH game, broad classes of nonlocal cooperative games have been proposed. 
These involve, in their simplest formulation, two spatially separated parties performing operations on some shared resource system, while being forbidden from communicating. 
Moreover, there also exist games cast in different setups that show the same classical and quantum bounds as the CHSH game  \cite{Ambainis1999,Spekkens2009,Fritz2010,Henaut2018}. These do not involve two spatially separated parties performing local operations on the corresponding systems, but only an input resource system subjected to an ordered sequence of transformations or measurements. We refer to this kind of games as \textit{monopartite sequential games}.
Obviously, in these cases, the Tsirelson bound cannot be read as a quantifier of the nonlocality of quantum theory and it becomes natural to wonder about the nature of the source of the quantum-over-classical advantage in such games, and whether there is a deep connection between these protocols and nonlocal games.

In this work we address these questions for particular classes of nonlocal and monopartite sequential games that include various known games studied in the literature. More precisely, we focus on a  subclass of nonlocal games -- XOR games \cite{cleve_consequences_2010} -- where the goal is to satisfy the condition that a (possibly nonlinear) function of the inputs equals the XOR (or sum modulo $2$) of the output bits.
Via Theorem \ref{theorem}, we relate a subset of these to a class of monopartite sequential games -- the XOR* games --
where the goal is to obtain, as single output, the (possibly nonlinear) function of the inputs.
The already mentioned CHSH game is the most famous example of an XOR game, but other examples within this class exist in the literature, such as the EAOS game \cite{brukner_entanglement-assisted_2006, faleiro_quantum_2019}  where  Alice and Bob have to compute a Kronecker-delta function of the input trits.  Regarding XOR* games, examples are  the  $2\rightarrow 1$ quantum random access codes (QRACs)\cite{Ambainis1999}, $2\rightarrow 1$ parity oblivious multiplexing (POM) \cite{Spekkens2009} and CHSH*~\cite{Henaut2018} games.
By virtue of our categorisation, the CHSH* game emerges as the natural XOR* version of the CHSH game and the QRACs read as the XOR* versions of an XOR game which is different from the CHSH game. This may sound surprising, as the $2\rightarrow 1$ QRAC is usually associated to the CHSH game in that it shares the same values for the Bell and Tsirelson bounds \cite{galvao}.  

Crucial in our work is Theorem \ref{theorem}: it allows us to connect the optimal quantum strategies for the classes of XOR and XOR* games. The theorem involves certain assumptions: bounding the cardinalities of the inputs (as a consequence of the use of a lemma due to Cleve et al. \cite{cleve_consequences_2010}) and restricting the XOR* games to those involving only two-dimensional (2D) resource systems. These constraints do not prevent the result to apply to a number of known XOR and XOR* games present in the literature. We describe the main ones in subsection \ref{Section_Examples} and, as a consequence of our framework, we also define the XOR versions of known XOR* games such as the $2\rightarrow 1$ QRAC \cite{Ambainis1999}, the binary output Torpedo game \cite{Emeriau2020}, and the Gallego et al. \cite{Gallego2010} dimensional witness.    
Theorem \ref{theorem} also assumes that the strategies of the XOR* games involve only reversible gates. We show how this assumption is strictly necessary for the theorem to hold by providing an example of XOR* game where the use of (irreversible) reset gates can create a quantum-classical performance gap that does not arise if considering only reversible gates. We refer to this feature as \textit{reset-induced gap activation}.
The mapping described in Theorem \ref{theorem} can be used to show that preparation contextuality \cite{spekkens_contextuality_2005} is a resource for outperforming classical strategies in certain XOR* games (see \ref{appendixB}). The reason is that the established proofs of nonlocality as a resource for outperforming classical strategies in XOR nonlocal games can be mapped to proofs of preparation contextuality in monopartite sequential games (interpreted as prepare and measure scenarios) via the mapping of the theorem. However, this argument does not apply to all the XOR* games that mimic the quantum over classical computational advantages of the XOR games and, crucially, does not apply to what we define as \textit{dual} XOR* games of the XOR games under consideration. For example, it does not apply to the CHSH* game. We provide a thorough analysis of this fact and of why other typical candidates employed to explain the origin of the quantum advantage do not work in the case of XOR* games. We conclude by outlining a proposal to develop a new notion of nonclassicality in the spirit of generalized contextuality that applies to these games.

The remainder of the paper is structured as follows. In section \ref{sec:games} we describe nonlocal games and monopartite sequential games, focusing specifically on XOR and XOR* games. In section \ref{Relating} we relate XOR and XOR* games via Theorem \ref{theorem}, which establishes a mapping between these two classes. In particular, we prove Theorem \ref{theorem} using a diagrammatic approach according to the formalism of process theories \cite{coecke2018picturing}, and we give various examples of games within these classes. In section \ref{Section_advantage} we discuss why the typical features employed to explain the quantum computational advantage in information processing tasks do not apply in the case of XOR* games and we advance a proposal to overcome this problem. 
In section \ref{conclusions} we discuss the significance of the results and we outline future challenges. Finally, we relegate all proofs of the stated lemmas to \ref{appendixA}.

\section{XOR nonlocal games and XOR*  sequential games}
\label{sec:games}

A  \textit{cooperative game} consists of a protocol involving some number of agents (players) that cooperate to perform a certain task. In order to win the game the players have to come up with a strategy to produce outputs which successfully compute a task function of the inputs provided by a referee. We formulate the winning condition as a \textit{predicate} that equates the task function with a function of the outputs.
A game is also specified by the \textit{setup} it is played in. The setup is defined by a list of elements that compose the protocol, \textit{e.g.}, the kind and number of resource systems involved, the number and cardinalities of inputs and outputs, and the number and sequential order of operations like transformations or measurements (controlled by the players). 
The list of actions that each player has to adopt given particular values of the inputs defines a \textit{strategy} for the game. Each strategy has to obey possible \textit{restrictions} that, given the specification of the setup, the protocol posits. 
The typical example of such a restriction, in the context of nonlocal games where the players are spatially separated and cannot communicate, is the principle of no-signaling. Notice that these restrictions are theory-independent, in the sense that they must be obeyed independently of the precise operational theory (\textit{e.g.}, classical or quantum theory) that the strategies of the players are formulated in.\footnote{The term ``theory-independent'' is standard in the literature of device-independent protocols (see for example \cite{Tavakoli2022informationally}). With this it is meant that the restrictions have to be obeyed by all the operational theories that are under comparison (\textit{e.g.}, classical, quantum and GPTs).}

Finally, a game, in general, can also be specified by extra constraints that do not emerge from the setup, but are \textit{artificial constraints} that are imposed on the set of possible strategies. Examples of these are the parity obliviousness in the parity oblivious multiplexing game \cite{Spekkens2009}, the dimensional constraint in quantum random access codes \cite{Wiesner1983,Ambainis1999} and, more recently, the information restriction on the communication channel in the protocols considered in \cite{Tavakoli2022informationally}. 
The purpose of introducing these artificial constraints is to allow for an interesting categorisation of the strategies, \textit{i.e.}, to create a performance gap between distinct operational theories. In other words, thanks to these constraints, it is possible to separate the performances of the strategies depending on the operational theory they are formulated in. In particular, they allow to observe performance gaps between classical and quantum strategies. Given some game \(G\), such a gap is characterised by a positive \textit{quantum-over-classical advantage},\footnote{In the following, we will refer to the quantum-over-classical advantage also as ``quantum-classical gap'' or ``performance gap''.} $\Omega(G) := \omega_q(G) - \omega_c(G) > 0$,  where $\omega_q(G)$ and $\omega_c(G)$ denote the highest probabilities to win $G$ by means of quantum and a classical strategy, respectively called the \textit{quantum value} and \textit{classical} \textit{value} of the game. Notice that when we say ``classical'' and ``quantum'' strategies we mean strategies that employ systems that are prepared, transformed and measured according to the rules of the respective theory. Furthermore, when referring  to the quantum-over-classical advantage and/or the classical and quantum values for games bearing  artificial constraints, we will explicitly write the type of constraint in the notation. For instance,  XOR* games restricted to two-dimensional resources and reversible operations, which we will extensively deal with later on, have their quantum-over-classical advantage written as  $\Omega(\textup{XOR}^*|^{Rev}_{2D}) := \omega_q(\textup{XOR}^*|^{Rev}_{2D}) - \omega_c(\textup{XOR}^*|^{Rev}_{2D}).$ 

In summary, a game is formally defined by a specification of a setup, predicate, restrictions and possible artificial constraints. In this section we are interested in two broad classes of games: \textit{nonlocal games} and \textit{monopartite sequential games}, which we now define in the case of two players (figure \ref{fig:games_diagrams}).

\begin{definition}[\textbf{Two-player Nonlocal game}] 
\label{nonlocalgame}
\end{definition}
\begin{itemize}
    \item \textit{Setup}:  the two players,  Alice and Bob, receive inputs \(s\in\mathcal{S}\) and \(t\in\mathcal{T}\), respectively, sampled from a known probability distribution $p(s,t)$, where $\mathcal{S}$ and $\mathcal{T}$ denote tuples of dits of finite dimension. Before the game starts, Alice and Bob can agree on a strategy, while they cannot communicate after the game starts. Such a strategy could involve using a shared resource system on which they perform local operations (transformations and measurements) controlled by the values of their respective inputs. Each player has to provide an output, denoted with  \(a\in\mathcal{A}\) for Alice and \(b\in\mathcal{B}\) for Bob, where $\mathcal{A}$ and $\mathcal{B}$ denote tuples of dits of finite dimension.    
 
    \item \textit{Restriction}: the restriction of the setup is the no-signalling condition, that prevents Alice and Bob -- in spatially separated locations -- from communicating during the game, \begin{equation}
    p(a|st) =  p(a|s); \;
   p(b|st) =  p(b|t).
 \end{equation}
    \item \textit{Predicate}:
    \begin{equation}
 W_{g|f}(a, b| s,t) =\left\{
\begin{array}{ll}
      1 \hspace{.1cm}\textup{if }\hspace{.1cm} f(s,t)=g(a,b)\\
      0 \hspace{.1cm}\textup{otherwise}.\\
\end{array}
\right.    
\end{equation}{}
\end{itemize}
The function $f(s,t)$ is what we called the ``task function'' at the beginning of the section.\\


\begin{definition}[\textbf{Two-player Monopartite Sequential game}]
\label{sequentialgame}
\end{definition}

\begin{itemize}
    \item \textit{Setup}: the two players,   Alice and Bob, receive inputs \(s\in\mathcal{S}\) and \(t\in\mathcal{T}\), respectively, sampled from a known probability distribution $p(s,t)$, where $\mathcal{S}$ and $\mathcal{T}$ denote tuples of dits of finite dimension.  They are ordered in sequence, say Alice before Bob ( denoted as ``$A<B$''). Alice is given a single resource system on which she applies one operation controlled by the value of the input $s$ she receives. The system is forwarded to Bob on which he also applies one operation controlled by the value of the input $t$ he receives. At the end the resource system is subjected to a fixed measurement. In this way it is subjected to a ordered sequence of operations (transformations or measurements) controlled by the inputs, and the strategy the players agree on consists of choosing the resource system, the controlled operations and the measurement.  The outcome $m\in\mathcal{M}$ of the measurement is the output of the game, where $\mathcal{M}$ denotes a tuple of dits of finite dimension. 

\item \textit{Restriction}: analogously to the principle of  no-signaling in nonlocal games, the physical restriction on the game is here that the communication between the players cannot violate the causal structure imposed by the setup. This means that Bob cannot signal backwards to Alice. This is called the principle of weak-causality \cite{Pegg}  and can be written as follows, 
\begin{equation}
    \label{No-signalling-from-the-future}
    p(a|s,t) =  p(a|s),
\end{equation}
where $a$ is the variable associated with Alice's choice of operation.
\item \textit{Artificial constraints}:  They might be of various types. For example, reversibility of operations, dimensional constraints on the resource, parity obliviousness, informational restriction of the channels, etc.

    \item \textit{Predicate}: \begin{equation}
 W_f(m| s,t) =\left\{
\begin{array}{ll}
      1 \hspace{.1cm}\textup{if }\hspace{.1cm} f(s,t) =m \\
      0 \hspace{.1cm}\textup{otherwise}.\\
\end{array}
\right.    
\end{equation}{}
\end{itemize}

Again, the function $f(s,t)$ is what we called the ``task function'' at the beginning of the section. 
  \begin{center}
  \begin{figure}[h!]
    \centering
      \includegraphics[scale=0.35]{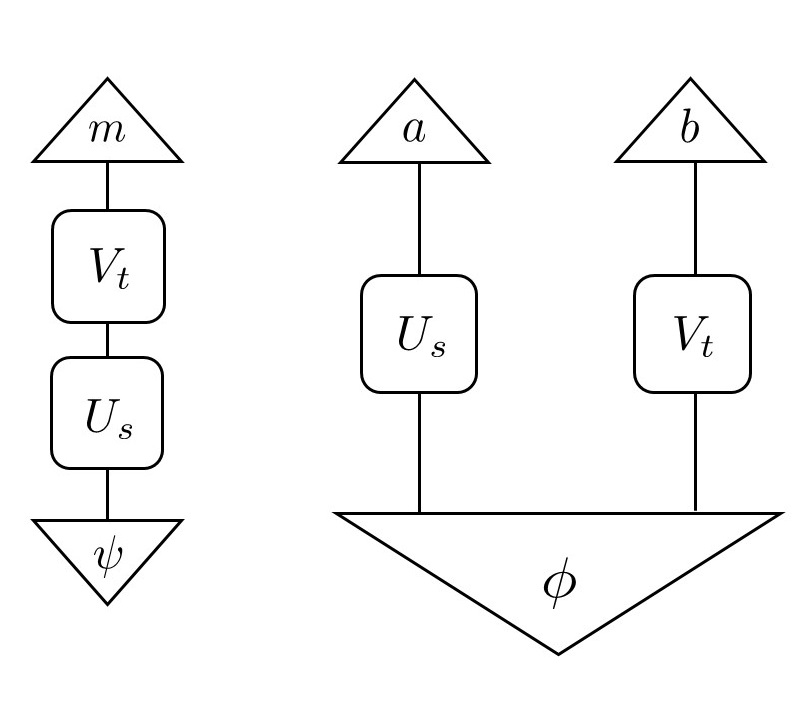}
     \caption{Setup of monopartite sequential games (left) and nonlocal games (right). In the former case the binary outcome $m$ is obtained from a unique measurement after two controlled operations $U_s$ and $V_t$ performed on a monopartite resource system whose state is denoted with $\psi$. In the latter case the outcomes $a,b$ are obtained via two spatially separated measurements, which are implemented after two controlled operations $U_s$ and $V_t$ have been performed on a bipartite resource system denoted with $\phi$. In addition to illustrating the setups of the games just described, the figure above can be read as a process diagram composed of systems (wires) and processes (boxes), where \textit{states} are special processes with trivial inputs (upside-down triangles), and \textit{effects} special processes with trivial outputs (triangles).  We will use the language of process theories later on when proving our main result.}
     \label{fig:games_diagrams}
\end{figure}

 \end{center}

A widely known class of two-player nonlocal games are the XOR games. First introduced in \cite{cleve_consequences_2010}, they have been studied in the context of multi-interactive proof systems with entanglement.
XOR games form a subset within the class of \textit{binary} nonlocal games, which are nonlocal games where the players output bits (even though the inputs do not necessarily need to be bits). An XOR game further restricts such a class by requiring, as the winning condition, that the particular task function must equate the parity (or XOR) of the output bits. The CHSH game is one example of an XOR game, where the winning condition states that the parity of the output bits must equal the product of the input bits. Other XOR games might consider different input sets, and distinct functions over such inputs; for instance, in the EAOS game \cite{faleiro_quantum_2019}  the function is a Kronecker-delta of the input trits.

Regarding monopartite sequential games, the precise subset that we will focus on, are XOR* games, where the winning condition is that the particular task function must equate the output bit. Our motivation for focusing on these specific games is to identify the class of monopartite sequential games that reproduce the same classical, quantum bounds present in XOR games, in analogy to the CHSH* and CHSH example first studied in \cite{Henaut2018}.\footnote{We refer the reader to \cite{Henaut2018} for the details of how to obtain the classical and quantum bounds and for the description of the optimal strategies in the CHSH* game.} This fact also motivates the name ``XOR*''. Following this approach we will connect the two classes of games via a theorem that relates their strategies and we will clarify the relationship between various examples of games that have been extensively studied and whose performances manifest the same bounds. We now characterise XOR and XOR* games in agreement with the previous definitions of nonlocal games (definition \ref{nonlocalgame}) and monopartite sequential games (definition \ref{sequentialgame}).\\

\begin{definition}[\textbf{XOR games}]
\label{XOR games}
\end{definition}
\begin{itemize}
    \item \textit{Setup}: Two-player nonlocal game setup, where the outputs are one bit for each player, \textit{i.e.}, $\mathcal{A}=\mathbb{Z}_2$ and $\mathcal{B}=\mathbb{Z}_2$. 
    \item \textit{Restriction}: No-signalling condition,  \begin{equation}
    p(a|st) =  p(a|s); \;
   p(b|st) =  p(b|t).
   \end{equation}
    \item \textit{Predicate}: \begin{equation}
 W_f(a\oplus b| s,t) =\left\{
\begin{array}{ll}
      1 \hspace{.1cm}\textup{if }\hspace{.1cm}f(s,t)=a\oplus b\\
      0 \hspace{.1cm}\textup{otherwise}.\\
\end{array}
\right.    
\end{equation}{}
\end{itemize}

\begin{definition}[\textbf{XOR* games}]
\label{XOR* games}
\end{definition}
\begin{itemize}
    \item \textit{Setup}: Two-player monopartite sequential game setup, where the output is a bit, \textit{i.e.}, $\mathcal{M}=\mathbb{Z}_2$. .
    
\item \textit{Restriction}: Weak-causality,
        \begin{equation}
    p(a|s,t) =  p(a|s);
\end{equation}
where $a$ is a classical variable labeling the state after Alice's operation.
\item \textit{Artificial constraints}: Various types. For example, reversible operations, dimensional constraints on the resource, parity obliviousness, etc.
    \item \textit{Predicate}: \begin{equation}
 W_f(m| s,t) =\left\{
\begin{array}{ll}
      1, \hspace{.1cm}\textup{if }\hspace{.1cm}f(s,t)=m\\
      0, \hspace{.1cm}\textup{otherwise}.\\
\end{array}
\right.    
\end{equation}{}
\end{itemize}
  
Notice how XOR* games are nothing but \textit{binary} monopartite sequential games.

\section{Relating XOR and XOR* games}
\label{Relating}

In this section we connect XOR and XOR* games. In order to do so, we impose constraints on the cardinality of inputs, as well as on the dimensionality of the resource system and the type of transformations in the XOR* games. In the following, with the aid of known results about XOR games and extending a proof contained in \cite{Henaut2018}, we first prove our main theorem. We then show that both assumptions of reversibility  of transformations and bi-dimensionality of the resource system in XOR* games are crucial for the theorem to hold by providing counter examples where such assumptions do not hold and the mapping of the theorem breaks down. 

\subsection{Main result}

We start by focusing on a subset of XOR games, which we call \textit{ebit} XOR games. They are XOR games for which the Tsirelson bound is obtained with strategies using one Bell state, that is, exploiting what is called an \textit{ebit}---a single \textit{bit} of shared bipartite entanglement \cite{wilde_2017}. We now state two lemmas concerning ebit XOR games that will be useful to prove our main theorem, Theorem \ref{theorem}.

\begin{lemma}
\label{lemma_cleve}
   Any XOR game where \(\textup{min}(|\mathcal{S}|,|\mathcal{T}|)\leq 4\) is an ebit-XOR game.
\end{lemma}
\begin{proof}
    From Theorem 10 of Cleve et al. \cite{cleve_consequences_2010},  we know that \(\lceil{m/2}\rceil\) qubits in a maximally entangled state, for  \(m=min (|\mathcal{S}|,|\mathcal{T}|)\), are sufficient to implement an optimal quantum strategy in an XOR game.  Consequently, a Bell state (a maximally entangled state of two qubits) is sufficient to implement an optimal quantum strategy when  \(min (|\mathcal{S}|,|\mathcal{T}|)\leq 4\).
\end{proof}

\begin{lemma}[V.1,\cite{mansfield2017consequences}]
Given inputs $(s,t)$ for Alice and Bob specifying their respective projective single-qubit measurements, $A_{a|s},B_{b|t}$, applied, respectively, to the first and second qubit of the two-qubit system in the Bell state $\ket{\phi^+}=\frac{1}{\sqrt{2}}\ket{00}+\ket{11}$, the probability \(p(a,b|s,t)\) of  obtaining output bits \((a,b)\)  is such  that
\(
     p(0,0\vert s,t) = p(1,1 \vert s,t) \) and \( p(0,1\vert s,t) = p(1,0 \vert s,t)
\).
 \label{lemma_symmetry}
 \end{lemma}
 \begin{proof}
According to the Born rule, the probability \(p(a,b|s,t)\) is given by $p(a,b|s,t)= \textrm{Tr}[(A_{a|s}\otimes B_{b|t})\ketbra{\phi^+}{\phi^+}]$, where  $A_{a|s}, B_{b|t}$ are given by $\{A_{0|s}= \ketbra{0_s}{0_s}, A_{1|s} = \ketbra{1_s}{1_s}\}$ and $\{B_{0|t} = \ketbra{0_t}{0_t}, B_{1|t}= \ketbra{1_t}{1_t}\}$, respectively, with 
 \begin{align}
     \ket{0_i} &= \cos\frac {\theta_i}2 \ket{0} + e^{i \phi_i} \sin\frac {\theta_i}2 \ket{1}\\
     \ket{1_i} &= \sin\frac {\theta_i}2 \ket{0} - e^{-i \phi_i} \cos\frac {\theta_i}2 \ket{1},
 \end{align}
where $\theta_i$ and $\phi_i$ are the angles defining the Bloch representation of the vectors, with $i\in\{s,t\}.$
Given that  $ \braket{0_s 0_t}{\phi^+} =  e^{-i(\phi_s + \phi_t)} \braket{1_s 1_t}{\phi^+}$ and $\braket{0_s 1_t}{\phi^+} =  -e^{-i(\phi_s - \phi_t)} \braket{1_s 0_t}{\phi^+}$, we obtain \begin{equation*}p(0,0\vert s,t) =|\braket{0_s 0_t}{\phi^+}|^2 =|\braket{1_s 1_t}{\phi^+}|^2 = p(1,1 \vert s,t) \end{equation*} and \begin{equation*} p(0,1\vert s,t) = |\braket{0_s 1_t}{\phi^+}|^2 = |\braket{1_s 0_t}{\phi^+}|^2 = p(1,0 \vert s,t).\end{equation*}

\end{proof}

We now state and prove our main theorem relating XOR and XOR* games. Most of the proof consists of showing that for any strategy employing a Bell state in an XOR game there exists a two-dimensional reversible strategy -- that is, a strategy involving a resource system which is two-dimensional (\textit{e.g.}, a bit or a qubit) and controlled gates that are reversible transformations -- in the corresponding XOR* game achieving the same performance, and vice versa. We denote the equivalence between the two quantum-over-classical computational advantages for the two classes of games as \(\Omega (\textup{XOR}) = \Omega(\textup{XOR}^*|^{Rev}_{2D})\).  In the proof we adopt a diagrammatic approach  (see Figure \ref{fig:slidingproof}) exploiting the sliding rule (Proposition 4.30, \cite{coecke2018picturing}), which is standard in the formulation of quantum mechanics as a process theory \cite{coecke2018picturing}. Quantum mechanics as a process theory inherits the structure of a symmetric monoidal category, and so can be represented by a theory of systems (wires) and processes (boxes) admitting sequential and parallel composition. \textit{States} and \textit{effects} are special types of processes represented as boxes with trivial inputs and outputs, respectively. Cup or cap shaped wire deformations are interpreted as Bell states and effects, which are fundamental for establishing a duality between processes and bipartite states \cite{JAMIOLKOWSKI, CHOI}.

\begin{theorem}
\label{theorem}
Consider two sets $\mathcal{S}$ and $\mathcal{T}$ where \(\textup{min}(|\mathcal{S}|,|\mathcal{T}|)\leq 4\).  For every XOR game with output bits \(a\) and \(b\), whose inputs $s,t$ are elements of the sets $\mathcal{S}$ and $\mathcal{T}$, there exists an XOR* game with a two-dimensional resource system and involving reversible gates that has the same inputs $s\in\mathcal{S}$ and $t\in\mathcal{T}$ and output bit \(m = a\oplus b\) such that the XOR and XOR* games have the same classical and quantum bounds. The converse implication also holds. As a consequence, the games  show the same quantum-over-classical
advantage, \textit{i.e.}, \(\Omega (\textup{XOR}) = \Omega(\textup{XOR}^*|^{Rev}_{2D})\). 
\end{theorem}

\begin{proof}
We start by establishing the connection between the quantum strategies in XOR* games and \textit{ebit}-XOR games following a diagrammatic approach (see figure \ref{fig:slidingproof}), explained below. 

\begin{figure*}
    \centering
     \includegraphics[width=0.9\textwidth]{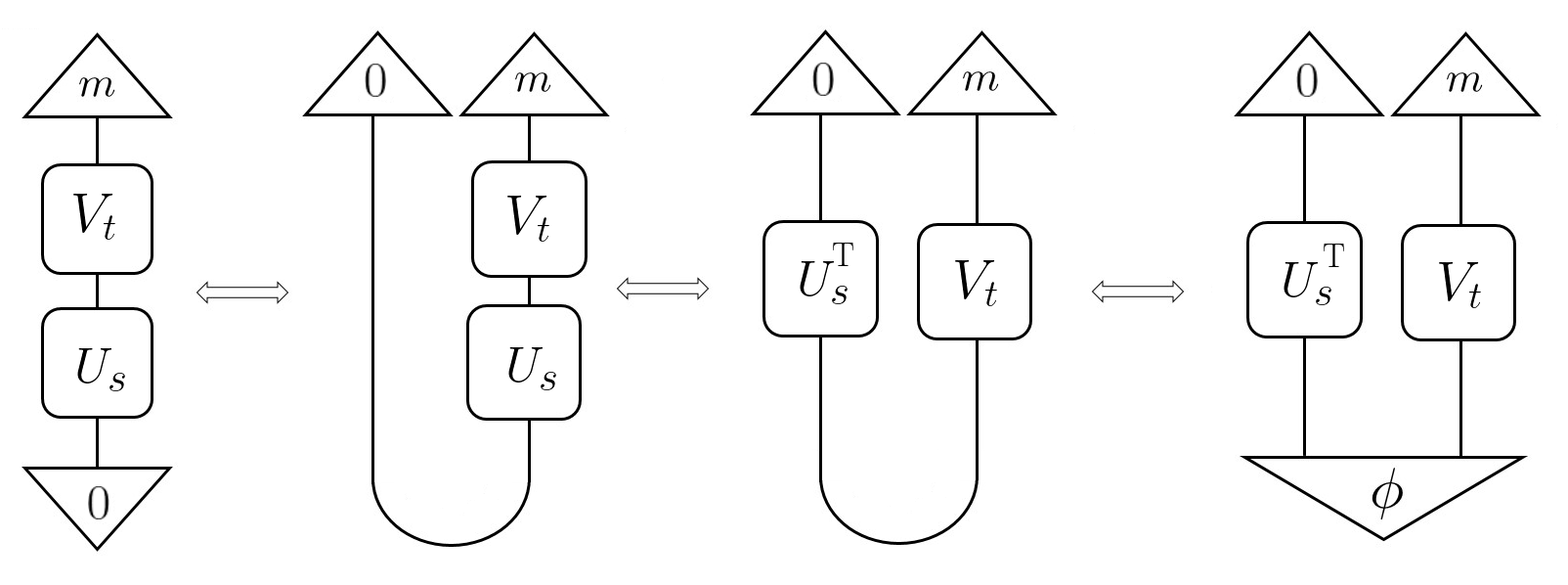}
    \caption{Diagrammatic representation of the mapping between the set of quantum strategies for ebit XOR games and bi-dimensional quantum reversible strategies for XOR* games.}
    \label{fig:slidingproof}
\end{figure*}

Let us consider the diagram on the left side of figure \ref{fig:slidingproof},\footnote{Although we are given a description of the diagrammatic proof left-to-right, because the sliding rule is symmetric the proof can alternatively be read from right-to-left, which is why the result holds in both directions, from the nonlocal to the monopartite sequential games, and vice versa.} that corresponds to the probability $p(m|s,t)$ associated with the setup of the XOR* game, which, for now, we assume it involves a qubit as resource system. Without loss of generality, we fix the state of the resource system to be $\ket{0}$ (choosing any other state would just correspond to implement an extra unitary in the controlled transformations). 
\begin{itemize}
    \item The first equivalence comes from the fact that a \textit{cup} can always be introduced in a diagram, through which we slide the \textit{state} \(\ket{0}\), turning it to an \textit{effect} (associated to the measurement element \(\ketbra{0}{0}\)). 
    \item The second equivalence comes from the fact that, according to the sliding rule (Proposition 4.30 \cite{coecke2018picturing}), Alice's unitary transformation can be sled through the cup, becoming its transpose in the process.
    \item The third equivalence is essentially the identification of the \textit{cup} with a Bell state. In this way we can now read the last diagram on the right side of figure \ref{fig:slidingproof} as describing a nonlocal game being played with such a Bell state.
\end{itemize}

 The diagrammatic equivalence in Figure \ref{fig:slidingproof} only works for transformations $U_s,V_t$ that are reversible and act on two-dimensional systems, and it shows that the probability $p(m|s,t)$ associated with the XOR* game is equivalent to the probability $p(a=0,b=m|s,t)$ associated with the \textit{ebit}-XOR game. We can extend this equivalence to any $p(a,b|s,t)$, with values of $a$ different from $0$, because, for XOR games, via Lemma \ref{lemma_symmetry}, we have that \(p(0,0\vert s,t) = p(1,1 \vert s,t)\) and  \( p(0,1\vert s,t) = p(1,0 \vert s,t) \).
 
 As shown by the previous diagram,  there is a bijective mapping between the sets of strategies using at most 1-\textit{ebit} resources in XOR games and the set of two-dimensional reversible strategies in XOR* games. This means that, for any (classical or quantum) strategy one can produce in one game, a corresponding strategy can be produced in the other achieving the same performance, and vice versa.
Furthermore, since in ebit-XOR games  the Tsirelson-bound is necessarily achievable by at most 1-ebit strategies, and from Lemma \ref{lemma_cleve} we know that any XOR game where \(\textup{min}(|\mathcal{S}|,|\mathcal{T}|)\leq 4\) is of this kind, then we are guaranteed that for games whose input sets' cardinality does not exceed 4 the quantum optimal strategies in both games provide the same values.  

Let us now consider the classical case. Since optimal classical strategies can be understood as deterministic assignments from inputs to outputs \cite{cleve_consequences_2010}, and the input and output sets of a pair of XOR  and XOR* games are by definition equal,\footnote{The output set in the case of XOR games is intended to be the set whose elements are the XOR of Alice's and Bob's output elements.} the set of reversible functions for both the XOR games (about which irreversible strategies cannot outperform reversible ones \footnote{This fact can be deduced from  Proposition 2 in \cite{cleve_consequences_2010}, whose consequence is that optimal strategies for XOR games are guaranteed by unitary operations in the quantum case, and reversible functions in the classical case.}) and XOR* games coincide, as do their Bell bounds.

\end{proof}

For those unfamiliar with the formalism of process theories, the proof can alternatively be stated with the standard quantum mechanical formalism. In particular, the sliding rule used in the diagrammatic proof above corresponds to a standard teleportation protocol, where by exploiting shared entanglement Alice teleports her gate to Bob's wing.\footnote{For an explicit circuit diagram, see  Figure 2 of \cite{Henaut2018}.}

\subsection{On the assumptions of the theorem}

Theorem \ref{theorem} crucially relies on three assumptions: an upper bound on the minimum cardinalities of the input sets, the reversibility of the gates and the two-dimensionality of the resource system in the XOR* games.
The first assumption imposes a constraint on the XOR side of the proof and is a direct consequence of Lemma \ref{lemma_cleve}. It guarantees that for XOR games respecting the bound on the cardinalities the quantum bound can be achieved with a strategy involving a Bell pair, which is crucial for establishing the mapping  in the direction from the XOR to the XOR* games. Notice that the assumption is sufficient but not necessary for guaranteeing that XOR games can be optimally won with a strategy employing a Bell pair, \textit{i.e.}, there could exist ebit XOR games violating the assumption. 

The remaining assumptions (reversibility and bi-dimensionality of the resource systems) impose constraints on the XOR* side of the proof, and guarantee that the map holds from the XOR* to the XOR games. These assumptions are necessary for establishing the mapping in the theorem, indeed there exist XOR* games for which increasing the dimension of the resource system or using irreversible gates allow to obtain higher classical and quantum bounds, while leaving unchanged the bounds in the corresponding XOR game. Notice how this does not mean that \textit{the difference} between the quantum and classical bounds can increase by dropping any of the two assumptions. Hence, given that one may just be interested in obtaining the highest quantum-over-classical advantage, it is not obvious whether increasing the dimension of the resource system or adopting irreversible gates can serve for such a purpose. For example, in the CHSH* game \cite{Henaut2018} it is shown that either considering a Reset gate (irreversible transformation) or going to dimension three, increases both the classical and quantum bounds to the value 1, thus decreasing the quantum-classical performance gap.
In what follows, we show that not all XOR* games are like the CHSH* game, in that there exist examples whose quantum-classical gap can increase by introducing resource systems of dimension greater than two or Reset gates. 

\paragraph{Increasing the dimensionality of the resource system --} An XOR* game showing an increase in the quantum-classical performance gap due to the use of a resource system with dimension greater than two is the binary output Torpedo game introduced in Emeriau et al. \cite{Emeriau2020}. The latter is an XOR* game with 3 inputs for Alice and 4 for Bob (see Example 5 in subsection \ref{Section_Examples}). It has been shown in \cite{Emeriau2020} with numeric methods that the game manifests the following quantum-classical gaps: \(\Omega(\textup{Bit-Torp.}|_{2D}) = 0.039\) and \(\Omega(\textup{Bit-Torp.}|_{3D}) = 0.042\). The fact that the dual XOR version of the binary output Torpedo game has quantum-classical gap equal to the one of the binary output Torpedo game with bi-dimensional resource system is then enough to show that Theorem \ref{theorem} without the bi-dimensionality assumption of the resource system does not hold.

\paragraph{Employing irreversible transformations --} In \ref{appendixA} we provide an example of an XOR* game which shows how  irreversibility can increase the quantum-classical gap. Remarkably, this example demonstrates the potentially resourceful role of irreversibility for establishing a quantum-over-classical advantage in XOR* games: the game manifests zero quantum-classical gap when allowing only reversible gates, but it shows a positive quantum-classical gap when employing Reset gates. This demonstrates a phenomenon of \textit{reset-induced gap activation}, and shows how irreversibility can create a quantum-classical gap, rather than just increase an already existing one. It is not known whether this type of activation is also possible in the previous case where one increases the bi-dimensionality of the resource system.\\

\subsection{Examples}
\label{Section_Examples}

Theorem \ref{theorem} states that any two-player XOR (or XOR*) game defined over the inputs sets $\mathcal{S}$ and $\mathcal{T}$, respecting a cardinality upper bound \(min(\mathcal{S},\mathcal{T})\leq 4\), is guaranteed to have a natural XOR* (or XOR) version, exhibiting the same classical and quantum bounds, for strategies in the XOR* game with a two-dimensional resource system and involving reversible gates. Accordingly, a number of examples of XOR games have a natural XOR* counterpart exhibiting the same quantum-over-classical advantage. This was already known for the CHSH game whose counterpart is the CHSH* game, but it is also true for other instances of XOR games already present in the literature. On the other hand, protocols that belong to the class of XOR* games subjected to the appropriate artificial constraints also have a natural XOR variant, like the \(2\rightarrow 1\) QRAC. In our categorization, the latter has a natural XOR game counterpart which we call \(2 \rightarrow 1\) QRAC\(^{\oplus}\) (see Example 4). Notice how unlike other approaches which associate the \(2 \rightarrow 1\) QRAC with the CHSH game \cite{galvao}, our mapping connects the CHSH and QRAC not with each other but to the CHSH* and \(2 \rightarrow 1\) QRAC\(^{\oplus}\), respectively.

Because of the connection established by this theorem, we will call a couple of connected XOR and XOR* games as a \textit{dual pair} XOR/XOR*, specified by the input sets, corresponding probability distributions, and the shared task function \(f\). It is implicit that there is a unique output bit, \(m\), in the XOR* game, and two individual output bits \(a,b\) in the nonlocal game, such that they are related by \(m = a\oplus b\).  We now provide several examples of dual pairs of XOR/XOR* games.

\begin{example}[CHSH \cite{CHSH69}/CHSH* \cite{Henaut2018}]
A two-player XOR/XOR* game  with inputs \(s,t \in \{0,1\}\) where, \(\forall_{s,t}  \; p(s,t)={1}/{4},\) and 
 \begin{equation} f(s,t):= s\cdot t.
\end{equation}
The classical and quantum bounds are, respectively,
\begin{equation*}
\begin{matrix} 
\omega_c(\textup{CHSH}) = \omega_c(\textup{CHSH}^*|^{Rev}_{2D})= & 3/4\\ 
\omega_q(\textup{CHSH}) = \omega_q(\textup{CHSH}^*|^{Rev}_{2D})= & \cos^2(\pi/8).
\end{matrix}
\end{equation*}
\end{example}

\begin{example}[\(n\)-Odd Cycle \cite{cleve_consequences_2010}/\(n\)-Odd Cycle*]
A two-player dual pair XOR/XOR*  with inputs \(s \in \{1,2, ... ,n\}\) for odd \(n\geq 3\), and \(t \in \{s, \;s\oplus 1 \mod n\}\) where, \(\forall_{s,t}  \;\;\; p(s,t=s)={1}/{2},\;\; p(s,t= s\oplus1 )={1}/{2},\) and
 \begin{equation} f(s,t):=[s\oplus 1 = t \;(\textup{mod}\; n)],
\end{equation} where the notation \([s\oplus 1 = t \;(\textup{mod}\; n)]\) denotes the truth value of the proposition \(s\oplus 1 = t \;(\textup{mod}\; n)\). The classical and quantum bounds are, respectively,
\begin{equation*}
\begin{matrix} 
 \omega_c(n\textup{-OC}) = \omega_c(n\textup{-OC}^*|^{Rev}_{2D})  = & 1 -
1/2n\\ 
\omega_q(n\textup{-OC}) = \omega_q(n\textup{-OC}^*|^{Rev}_{2D}) = & \cos^2(\pi/4n).
\end{matrix}
\end{equation*}

\end{example}

\begin{example}[EAOS \cite{brukner_entanglement-assisted_2006,faleiro_quantum_2019}/EAOS* \cite{coiteux-roy_advantage_2020} ]
A two-player dual pair XOR/XOR* with inputs \(s,t \in \{1,2,3\}\) where, \(\forall_{s,t}  \; p(s,t)={1}/{9},\) and
\begin{equation} f(s,t):=  \overline{\delta(s,t)},
\end{equation} where \(\delta(s,t)\) is the Kronecker delta function evaluating to 1 only when \(s=t\) and 0 otherwise,  and \(\overline{\delta(s,t)}\) its negation.
The classical and quantum bounds are, respectively,
\begin{equation*}
\begin{matrix} 
 \omega_c(\textup{EAOS}) = \omega_c(\textup{EAOS}^*|^{Rev}_{2D})  =  & 7/9\\ 
\omega_q(\textup{EAOS}) = \omega_q(\textup{EAOS}^*|^{Rev}_{2D}) = & 5/6.
\end{matrix}
\end{equation*}

\end{example}

\begin{example}
[$2 \rightarrow 1$ QRAC\(^{\oplus}\) \cite{Chailloux2016}/ $2 \rightarrow 1$ QRAC \cite{Ambainis1999}]
A two-player dual pair XOR/XOR* with inputs $\textbf{s} = (s_0, s_1)$ for  \(s_0,s_1 \in \{0,1\}\), \(t \in \{0,1\}\) where, \(\forall_{s,t}  \; p(\textbf{s},t)={1}/{8},\) and
\begin{equation} f(\textbf{s},t):=  s_0\cdot (t\oplus 1) \oplus s_1\cdot t .
\end{equation}

The classical and quantum bounds are, respectively,
\begin{equation*}
\begin{matrix} 
 \omega_c(\textup{$2 \rightarrow 1$ QRAC\(^{\oplus}\)}) = \omega_c(\textup{$2 \rightarrow 1$ QRAC}|^{Rev}_{2D}) =  &  3/4\\ 
 \omega_q(\textup{$2 \rightarrow 1$ QRAC\(^{\oplus}\)}) = \omega_q(\textup{$2 \rightarrow 1$ QRAC}|^{Rev}_{2D})=  & \cos^2(\pi/8).
\end{matrix}
\end{equation*}
\end{example}

Notice how, according to our framework, the nonlocal XOR game associated to the $2 \rightarrow 1$ QRAC is not the CHSH game, but an XOR game, that we denote as $2 \rightarrow 1$ QRAC\(^{\oplus}\). This connection between the QRACs and nonlocal games had already been studied in \cite{Chailloux2016}, where the corresponding nonlocal version were denoted as INDEX games.

\begin{example}
[Bit-Torpedo\(^{\oplus}\)/ Bit-Torpedo \cite{Emeriau2020}] 
A two-player dual pair XOR/XOR* with inputs $\textbf{s} = (s_0, s_1)$ for  \(s_0,s_1 \in \{0,1\}\), \(t \in \{0,1,2\}\) where, \( \forall_{s,t}  \; p(\textbf{s},t)={1}/{12},\) and
\begin{multline}
     f(\textbf{s},t):= \left(s_0\cdot (t\oplus_3 1)\!\!\!\!\mod 2\right)\oplus \left(s_1\cdot t\!\!\!\!\mod 2\right) \oplus \\\left((s_0\oplus s_1)\cdot (t\oplus_3 2)\!\!\!\!\mod 2\right).
\end{multline}

The classical and quantum bounds are, respectively,
\begin{equation*}
\begin{matrix} 
\omega_c(\textup{\textup{Bit-Torp.}\(^{\oplus}\)}) = \omega_c(\textup{Bit-Torp.}|^{Rev}_{2D}) = &  3/4\\ 
 \omega_q(\textup{Bit-\textup{Torp.}\(^{\oplus}\)}) = \omega_q(\textup{Bit-Torp.}|^{Rev}_{2D})  \simeq  & 0.789.
\end{matrix}
\end{equation*}
\end{example}

Notice that we write ``\(\oplus_3\)'' to denote the sum modulo 3 and distinguish it from the usual XOR operation ``\(\oplus\)'' that denotes the sum modulo 2.  In comparison with Example 4, the task function for the binary output Torpedo game -- here denoted as  ``Bit-Torpedo'' -- can be seen to be a generalisation of the $2 \rightarrow 1$ QRAC, where, in addition of Bob possibly being asked to output either bit corresponding to Alice's inputs, he can be further asked to output their parity. Notice also that due to the mapping of Theorem \ref{theorem}, we are able to automatically give a nonlocal version of the binary output Torpedo game, that we denote as ``Bit-Torpedo\(^{\oplus}\)''.

\begin{example}
[GBHA-\(I_3^{\oplus}\)/ GBHA-\(I_3\) \cite{Gallego2010}] 
A two-player dual pair XOR/XOR*  with inputs \(s\in \{0,1,2\}\) and  \(t\in \{0,1\}\) where \(p(2,1)=0\) and  \( p(s,t)={1}/{5}\) otherwise, and 
 \begin{equation} f(s,t):=\Theta_{2}(s+t),
\end{equation} where $\Theta_{2}(s+t)$ is the Heaviside step-function evaluating to 1 when \(s+t \geq 2\), and 0 otherwise. The classical and quantum bounds are, respectively,
\begin{equation*}
\begin{matrix} 
\omega_c(\textup{GBHA-}I_3^{\oplus}) = \omega_c(\textup{GBHA-}I_3|^{Rev}_{2D}) = & 4/5\\ 
 \omega_q(\textup{GBHA-}I_3^{\oplus}) = \omega_q(\textup{GBHA-}I_3|^{Rev}_{2D}) \simeq   & 0.88.
\end{matrix}
\end{equation*}

\end{example}

The previous game was introduced in Gallego et al. \cite{Gallego2010} to serve as a dimensional witness (detecting, from the performance of the game, whether the dimension of the resource system is greater than two). The game is a straightforward generalization of the CHSH* game, where Alice instead of receiving a bit receives a trit.  Indeed, by reducing Alice's input cardinality from three to two, one eliminates the (2,0) and (2,1) input possibilities, the Heaviside function reduces to the AND operation (or product) of the input bits, and the game becomes a version of the CHSH* game.

\section{On the source of quantum computational advantage in XOR* games}
\label{Section_advantage}

Although for XOR games nonlocality has been shown to be a source of the quantum advantage \cite{B14}, it is still not clear which nonclassical feature plays the same role in XOR* games. Since such games employ monopartite computational resources, nonlocality can be immediately excluded as a possible candidate. Similarly, since these protocols only consider a final fixed measurement on a two level system, Kochen-Specker contextuality, which requires varied sets of commuting observables in order to be witnessed, and which in any case does not arise in Hilbert spaces of dimension two, can also be excluded as a possible candidate.

It is possible to identify a resource for computational advantage in at least one specific class of XOR* games. In communication tasks with appropriate parity constraints it is known that preparation contextuality provides a resource for the advantage \cite{Chaturvedi2017,Ambainis2019}, the golden example being the \(2 \rightarrow 1\) parity oblivious multiplexing task \cite{Spekkens2009}.
Can we extend these results to the whole class of XOR* games? One thing that we can do is to leverage the result of Theorem \ref{theorem} and the mapping it involves to show that XOR* games, if interpreted as prepare-and-measure scenarios, are powered by preparation contextuality. Indeed, we know that in XOR games nonlocality is a resource to outperform classical strategies and, in a (bipartite) Bell scenario like that of XOR games, any proof of nonlocality can be mapped to a proof of preparation contextuality in either wing of the experiment: in Bob's (Alice's) wing one considers the prepare and measure scenario where the preparation is specified by the state Alice (Bob) steers Bob’s (Alice's) system to, and the measurement is Bob's (Alice's) measurement. This steering process corresponds exactly to the mapping involved in Theorem~\ref{theorem}, which turns the Bell scenario into a prepare-and-measure scenario.
In \ref{appendixB} we provide the definitions of ontological models, preparation noncontextuality, and a detailed proof of the above claim.

There is an important consideration to bear in mind about the role of preparation contextuality in XOR* games, however.
The mapping of Theorem \ref{theorem}, via the above argument, does \textit{not} send proofs of nonlocality in a given XOR game to proofs of preparation contextuality in the dual XOR* game, but rather guarantees such a proof for an XOR* game in a scenario whose set of inputs has larger cardinality -- see \ref{appendixB}.
As a crucial consequence, we cannot prove preparation contextuality via the mapping above in the case of an XOR* game with binary inputs, $s,t\in\mathbb{Z}_2,$ as is the case for the CHSH* game. The CHSH* game indeed seems not to allow for a proof of preparation contextuality as it only involves four inputs in total, and yet four inputs cannot select four preparations and two measurements, which is known to be the minimum requirement for a proof of preparation contextuality \cite{pusey2018robust}.
Due to this observation, it cannot be said that the mapping from XOR to XOR* games lifts in an obvious way to a mapping from the source of advantage in XOR games to that in XOR* games.

In summary, when confronted with the appearance of a quantum-over-classical advantage in XOR* games, we are pressed to pursue new ways of  identifying the feature that corresponds to nonlocality in XOR games. Given that the only degrees of freedom in the XOR* games are in the choices of controlled transformations, it seems natural to focus on a property of transformations as the source of computational advantage. 
In this respect, let us consider the notion of sequential transformation contextuality developed in \cite{Mansfield2018}, that has indeed been proven to be necessary for quantum advantage in a broad class of information retrieval games (which are instances of monopartite sequential games) \cite{Emeriau2020} as well as the CHSH* game. 

A caveat to the latter results is that sequential transformation contextuality applies either (a) to cases where the ontological model is required to preserve the dimensional restriction defining the operational setting or (b) to cases where the ontological model is required to satisfy $\oplus L$-ontology, which essentially restricts transformations to being represented as sums modulo $2$ (this assumption is advocated for in cases where the operational setting is the one of a parity computer \cite{Mansfield2018}).
So, sequential transformation contextuality is only relevant relative to the extent to which these conditions hold.
While the conditions can be argued to be natural in certain informational or computational tasks, we cannot claim that they will always be reasonable assumptions in general.

Is assumption (a) -- rather than (b), which is unwarranted -- natural in the case of XOR* games? If one takes the notion of naturalness as the notion of Leibnizianity \cite{Schmid2021unscrambling} (also stated as a no fine-tuning principle \cite{CataniLeifer2020}) -- \textit{i.e.}, that the operational equivalences predicted in principle to hold by the theory are preserved at the ontological level -- which is at the base of the principle of generalized noncontextuality, assumption (a) is natural as long as the dimensional restriction is written in terms of operational quantities and is empirically verifiable. 
However, it is important to notice that it is the notion of sequential transformation noncontextuality, where contexts are sequences of transformations, that is not natural in this perspective because, in general, one cannot verify an operational equivalence. Indeed, it is not possible to isolate a transformation in order to perform process tomography on it since it manifests in a specific context or sequence of transformations. Therefore, if one wants to stick to the credentials of the principle of generalized noncontextuality one would also need to replace sequential transformation noncontextuality accordingly.
Rather than such credentials, the latter is motivated from a perspective in which compositionality is taken to be the notion of naturalness -- meaning that operational composition should be reflected at the ontological level.

Let us conclude by stressing that, whatever the notion of noncontextuality adopted, it seems unavoidable the need to impose restrictions on the ontological models associated with the setup of XOR* games, for otherwise no meaningful results can be obtained (clearly there always exists a noncontextual ontological model that can compute a function $f(s,t)=m$ -- and that is the one associated with a classical computer).
In particular, it would seem natural that the ontological models have to be restricted such that the operational features of the XOR* games are preserved therein, \textit{i.e.}, the transition matrices associated with the transformations are reversible (permutations in the discrete case) and the dimensionality of the ontological model is constrained according to the fact that the resource system is a two-level system.
Then, for one who subscribes to the credentials of Leibnizianity/no-fine tuning, the aim would be to adopt the notion of noncontextuality that exploits operational equivalences arising in the XOR* games. A possibility is that the operational equivalences arise in fact from the operational restrictions -- reversibility and the bi-dimensionality of the resource system -- defining the XOR* games, and so assuming noncontextuality would already take care of appropriately restricting the ontological models. 
Ultimately, it would be desirable to show that every time one implements a strategy that wins the game with a probability greater than the Bell bound, this is inconsistent with the naturally justified restrictions on the ontological models and the adopted notion of (transformation) noncontextuality. We leave the development of this project for future research.

\section{Conclusion}
\label{conclusions}
The CHSH* game and its relationship to the CHSH game were first studied in \cite{Henaut2018}. 
Although belonging to distinct setups---the CHSH game being a nonlocal game and the CHSH* game a monopartite sequential game---they show the same classical and quantum bounds and they can be connected via an explicit mapping.
In this work we generalised such mapping to a broader class of games. 
More precisely, we defined a new class of monopartite sequential games---the XOR* games---with the idea of providing a natural generalisation of the CHSH* game, in the same way XOR games are a generalisation of the CHSH game. 
In addition to providing this new categorization of nonlocal and monopartite sequential games, we proved a theorem stating that for two input sets whose cardinality does not exceed $4$, there is a mapping between dual pairs of XOR/XOR* games such that they manifest the same quantum and classical performance bounds, for strategies in the XOR* game involving two-dimensional resource systems and reversible gates. We further presented two examples of XOR* games to demonstrate how lifting the assumptions of reversibility of gates and bi-dimensionality of the resource system assumptions leads to the mapping no longer holding. 

To summarize our main contribution, we showed that under certain conditions the dual pair XOR/XOR* manifests the same quantum-over-classical advantage, and the dual pair CHSH/CHSH*  is a particular example of this. We also provided other examples of such dual pairs (namely, \(n\)-Odd Cycle/\(n\)-Odd Cycle* games, EAOS/EAOS* games,
$2 \rightarrow 1$ QRAC\(^{\oplus}\)/ $2 \rightarrow 1$ QRAC games,
Bit-Torpedo\(^{\oplus}\)/ Bit-Torpedo games).
In addition to this finding, we also pointed out how the mapping that we established in Theorem \ref{theorem}, supplemented with the fact that nonlocality is a resource for outperforming classical strategies in XOR games, provides a way of showing that preparation contextuality is a resource that powers certain XOR* games when treated as prepare and measure scenarios. However, we also highlighted the limitations of this claim, the main one being that it does not apply to the CHSH* game. 

 Our work is related to other works focused on monopartite sequential protocols that show the same performances of the CHSH game. In particular, protocols known as temporal CHSH scenarios use controlled measurements on a single system and are proven to manifest the same Bell and Tsirelson bounds as the CHSH game \cite{Fritz2010,Budroni2013,Budroni2014,Markiewicz2014,Le2017,Hoffman2018,Budroni2019,Spee2020,Spee2020simulating,Vieira2022,Mao2022}. However, their focus is mostly on characterizing the set of allowed correlations and to test the assumptions of macrorealism and non-invasiveness (Legget-Garg inequalities) \cite{LeggettGarg1985,Lapiedra2006,Avis2010}. Alternatively, in \cite{Brierley} the authors have aimed to provide a different characterisation of non-classicality in the temporal setting, other than the Legget-Garg type, by analysing the cost (in terms of channel capacity) of classically simulating temporal correlations. It should be noted that, although from their perspective  any temporal qubit correlations are trivially classically simulatable by means of implementation of Toner-Bacon  protocols \cite{TonerBacon}, this would require extra ancillary systems, and thus, it is forbidden within our XOR* setup since it contradicts the assumption of using a single monopartite resource.\footnote{The single monopartite assumption is closer to the ``restricted working memory" constraint considered in the recently introduced computation model proposed by researchers at IBM \cite{IBM}.} A work that is more focused on the characterisation of games and their artificial constraints is \cite{Tavakoli2017}, that shows how preparation contextuality is a resource for communication games with constraints akin to the parity obliviousness in \cite{Spekkens2009}. However, these games do not always belong to the set of XOR* games that we consider in our work and therefore do not obey in general the mapping that we show in Theorem \ref{theorem}. 
That said, both the communication games studied in \cite{Tavakoli2017} and our XOR* games are strictly related to nonlocal games. Namely, in \cite{Tavakoli2017} a construction is shown that allows one to define communication games from nonlocal games, which in spirit bears some resemblance with the mapping from our Theorem \ref{theorem}. It should be noted, however, that even for examples of games that fit into both frameworks the mappings differ. For instance, in our framework the CHSH game is mapped naturally to the CHSH* game, whereas in  \cite{Tavakoli2017} the CHSH game is mapped to the $2\rightarrow 1$ POM which, although an XOR* game, it is evidently different from the CHSH* game.
Our XOR* games are also related to the protocols in \cite{Gallego2010}, even if in there they are treated as prepare and measure scenarios and are specifically devised to provide dimensional witnesses \cite{Brunner2008,Brunner2013,Guhne2014,Li2018,Sohbi2021}.  Finally, more recently in \cite{Diviánszky2023} a class of inequalities for two-dimensional systems in the prepare-and-measure scenario were established by means of numerical methods. From the game perspective, a prepare-and-measure scenario dealing with two-dimensional resources can be made to correspond to a two-player sequential monopartite game under the constraint of reversibility and two-dimensional resources. As such, since those inequalities could be naturally interpreted as XOR* games (meaning that there would exist a unique bit outcome from the measurement), our mapping would automatically establish a natural XOR version. Crucially, the number of preparations and measurements considered for some of the inequalities explored in \cite{Diviánszky2023} can go as high as 70 settings for preparations and measurements, which would drastically break the upper-bound on the cardinality that guarantees the optimality of a Bell pair in the corresponding XOR games. This motivates once again the pursuit of an alternative which would not be so restrictive in this regard.

We conclude by listing a couple of open questions originating from our work. First, it would be interesting to further generalise Theorem \ref{theorem} by enlarging the classes of games to consider, in particular by finding a proof strategy that does not rely on Lemma \ref{lemma_cleve} from \cite{cleve_consequences_2010}. 
The approach of \cite{tavakoli_predicate} developed to find a non trivial nonlocal version of the three dimensional Torpedo game \cite{Emeriau2020} might provide a first intuition towards such generalization. They use the same construction with Wigner negative states and phase-point operators of \cite{Emeriau2020} and they devise a strategy to appropriately modify the predicate of some monopartite sequential game in a way that the classical and quantum bounds of a corresponding nonlocal game match those of the original sequential monopartite game. This would allow to generalize Theorem \ref{theorem} to input sets of arbitrary cardinalities.
Second, in terms of the resources for computational advantage in XOR* games, it is worth exploring whether properties of transformations can be proven to be necessary for the quantum-over-classical advantage. In particular, we have advocated for the development of a notion of transformation noncontextuality with restrictions on the ontological models endorsing the same credentials of generalized noncontextuality.

\vskip6pt

\enlargethispage{20pt}

 \section*{Acknowledgements}{LC and AP acknowledge support from the Einstein Research Unit `Perspectives of a Quantum Digital Transformation’. AP also acknowledges support from the Emmy Noether DFG grant No. 418294583, and from the BMWK via the project Qompiler. RF thanks Emmanuel Zambrini Cruzeiro and Flavien Hirsch for discussions about the reset-induced gap activation example, and  acknowledges funding from FCT/MCTES through national funds and when applicable EU funds under the project UIDB/50008/2020, and support of the QuantaGenomics project funded within the QuantERA II Programme that has received funding from the European Union’s Horizon 2020 research and innovation programme under Grant Agreement No 101017733, and with funding organisations, The Foundation for Science and Technology – FCT (QuantERA/0001/2021), Agence Nationale de la Recherche - ANR, and State Research Agency – AEI.
}


\vskip2pc

\bibliographystyle{ieeetr}
\bibliography{sample}

\section*{Appendix}
\appendix

 \section{Reset-induced gap activation in XOR* games}
\label{appendixA}

In this appendix we introduce an example of XOR* game whose quantum-classical gap can be activated by allowing irreversibility. First, before introducing the example, we state two salient features of the role of irreversibility in XOR* games.

\begin{remark}
For XOR* games restricted to bi-dimensional resources:
\begin{enumerate}
    \item   The only useful irreversible transformation is the Reset gate (as defined below).
    \item  The Reset gate is only useful when applied by the second player (Bob).
\end{enumerate}
\end{remark}

 \begin{definition}[Reset gate]
 \label{ResetGate}
The Reset gate $\mathrm{R}$ is represented by  the completely positive trace-preserving map, $\mathrm{R}(\rho) = \ketbra{0}$ for all possible qubit states \(\rho\).
 \end{definition}
The Reset gate can be interpreted as a particular case of a \textit{replacement channel} \cite{watrous_2018}, where the replaced state is always the \(\ketbra{0}\) pure state. Such channels represent the action of discarding the original state and replacing it with some other state. 

For classical strategies it is well known that deterministic functions between inputs and outputs suffice to achieve optimality \cite{cleve_consequences_2010}. That is, neither shared randomness nor local randomness are useful. Thus, we can only consider the four possible deterministic maps for bi-dimensional systems, that are: identity, NOT, reset to zero and reset to one. Among these, the only irreversible ones are the reset to zero or the reset to one.  The fact that implementing the Reset gate can only be useful for the second player is also straightforward to show. Since the first player receives the resource always initialized in a fixed state (classically, either zero or one), the effect of implementing a Reset gate is tantamount to implementing a reversible gate that, for the initialized fixed state input, produces the same output, e.g, for a state initialized at zero a reset to zero gives the same output as implementing an identity, and for a state initialized at one the reset to zero gives the same output as implementing a NOT gate. Thus, the Reset gate, being a constant function, can only be leveraged by the second player to discard the actions of the first player. 

Regarding the case of employing irreversibility in quantum strategies, one might wonder if there may exist extra operations (like a ``quantum Reset gate''), other than the previously defined Reset gate, that could be even  more useful than the Reset gate. If so, such operations should be purity preserving on the grounds that classical randomness is not beneficial, just like in the classical case. Then, such an hypothethical ``quantum Reset gate" would act as  $\mathrm{R}_q(\rho) = \sigma$ for all \(\rho\), where \(\sigma\) has the same degree of mixedness as \(\rho\). In particular, in the pure state case,  $\mathrm{R}_q(\ketbra{\psi}) = \ketbra{\phi}$. This, in turn, can be re-written as $\mathrm{R}_q(\ketbra{\psi}) = U_{\phi}\ketbra{0}U^{\dagger}_{\phi} = U_{\phi}\mathrm{R}(\ketbra{\psi})U^{\dagger}_{\phi} $, where $\ket{\phi}=U_{\phi}\ket{0},$ thus showing that such a ``quantum Reset gate" is nothing more than a composition of a regular Reset gate to $\ket{0}$ and a unitary evolution, for the appropriate unitary \(U_{\phi}\). As such, the Reset gate as defined in definition \ref{ResetGate} is the only type of irreversible operation that is useful to consider for XOR* games restricted to (classical or quantum) bi-dimensional resources.

\begin{example}
[Reset-induced gap activation (RA) XOR* game]
    Two-player XOR* game with inputs \(s,t\in \{0,1,2,3\}\) where  \( \forall_{s,t}  \; p(s,t)={1}/{16}\), and 
\( f(s,t):={\delta(s\cdot t, 0)} \oplus \delta(s\cdot t, 3)\). 
\end{example}

The previous example is by construction a legitimate XOR* game. Nevertheless,  under the assumption of reversibility and of two-dimensional systems the game's classical and quantum values are numerically computed to be identical, namely $\omega_q(\textup{RA-XOR}^*|^{Rev}_{2D})= \omega_c(\textup{RA-XOR}^*|^{Rev}_{2D}) = 13/16 = 0.8125$. This game represents then an interesting example where the quantum-classical gap is null, without being trivial in the sense that both classical and quantum strategies allow a winning probability of one. From the previous results on the role of the Reset gate, we are motivated to pursue a strategy where Bob uses such gates. Ideally, the Reset gate is implemented for input configurations whose evaluation by the task function only depends on Bob's input. Precisely, from the function of the previous example: \(\forall_s \; 
f(s,0):={\delta(0, 0)} \oplus \delta(0, 3) = 1\). This means that a strategy where Bob for input 0 resets to 1 always wins for all of Alice's inputs. The winning probability of an irreversible classical strategy implementing the Reset gate in this way can be easily numerically computed to be optimal and yields the classical value \(\omega_c(\textup{RA-XOR}^*|^{Irr}_{2D})=14/16 = 0.875\), which is higher than both the classical and quantum value in the reversible setting. 

Given the considerations above, it is now straightforward to construct an irreversible quantum strategy that outperforms the optimal classical irreversible strategy. It is enough to implement a Reset gate in the same way as before, where Bob for input 0 always resets to \(\ket{1}\), and to optimise for quantum reversible strategies in the remaining input configurations. This yields a winning probability of \(W_q(\textup{RA-XOR}^*|^{Irr}_{2D})\simeq 0.885\) \footnote{Numerical findings suggest that this winning
probability is the optimal quantum value in the irreversible setting. Notice that finding the optimal quantum value is not needed for showing Reset-induced gap activation.}, which means that the quantum-classical gap goes from zero, in the reversible setting, to a strictly positive value in the irreversible setting,  \(\Omega(\textup{RA-XOR}^*|^{Irr}_{2D})\geq 0.01 > 0 = \Omega(\textup{RA-XOR}^*|^{Rev}_{2D})\). This then shows an instance of quantum-classical gap  creation by employing  irreversibility in the form of the Reset gates, i.e. \textit{Reset-induced gap activation}.

 \section{Preparation contextuality in XOR* games}
\label{appendixB}
In this appendix we first provide the definition of preparation noncontextuality \cite{spekkens_contextuality_2005} and then prove how preparation noncontextuality in the prepare and measure scenario associated with any wing of the Bell scenario implies locality therein (this is an already known fact -- see for example section V. and appendix A of \cite{Yadavalli2020}
). The latter is equivalent to show that nonlocality in Bell scenario -- in our case defining an XOR game -- implies preparation contextuality in any wing of such scenario -- here defining an XOR* game.

The notion of preparation noncontextuality is defined in the framework of ontological models \cite{Harrigan2010}. The role of an ontological model for a given physical theory, like quantum theory, is to reproduce and provide an explanation of the statistics predicted by the theory. It does so by associating the physical state of the system under consideration at a given time -- the ontic state -- to a point $\lambda$ in a measurable set $\Lambda,$ and the experimental procedures -- classified in preparations, transformations and measurements -- to probability distributions on the ontic space $\Lambda.$ The ontological model reproduces the statistics of the theory by making use of the classical law of total probability. 
We are here interested in experimental procedures that correspond to lists of instructions to prepare a system in a given quantum state. Such a preparation procedure $P$ of a quantum state $\rho$ is represented by a probability distribution $\mu_P(\lambda)$ over the ontic space, $\mu_P:\Lambda\rightarrow \mathbb{R}$ such that $\int \mu_P(\lambda)d\lambda=1$ and $\mu_{P}(\lambda)\ge0 \;\; \forall \lambda\in\Lambda.$ 
We say that two preparation procedures, $P$ and $P'$, in a prepare and measure scenario are \textit{operationally equivalent}, $P\simeq P',$ if they provide the same statistics for any measurement, \textit{i.e.}, $p(m|P,M)=p(m|P',M)$ for every measurement $M$ and any outcome $m.$ A simple example of operationally equivalent preparations in quantum theory is given by any two decompositions of the completely mixed state $\rho=\mathbb{I}/2$ of a qubit, \textit{e.g.}, $P=1/2\ket{0}\bra{0} + 1/2\ket{1}\bra{1}$ and $P'=1/2\ket{+}\bra{+} + 1/2\ket{-}\bra{-}.$

We here focus on prepare and measure scenarios. An ontological model of a prepare and measure scenario in quantum theory is \emph{preparation noncontextual} if \begin{equation}\label{TransfNC}\mu_{P}(\lambda)=\mu_{P'}(\lambda) \;\;\; \forall \;P\simeq P'.\end{equation}

We can now show, following \cite{Yadavalli2020}, how preparation noncontextuality in the prepare and measure scenario on Bob's wing -- let us choose Bob's wing without loss of generality -- implies locality in the bipartite Bell scenario  \cite{Bell64,Bell75}. Let us stress that this ultimately means that a proof of nonlocality in an XOR game, cast in Bell scenario, implies, by virtue of our mapping, preparation contextuality in an XOR* game treated as a prepare and measure scenario.

We start by focusing on the operational equivalences induced by Alice on Bob's side. Whenever she performs a measurement, that we denote with the POVM $\{E_a^{(s)}\}_a$, and obtains the outcome associated with $E_a^{(s)}$ she steers, with probability $p(a|s)$, Bob's system to the state $\rho_{a|s}=\textrm{Tr}_A[(E_a^{(s)}\otimes \mathbb{I})\rho_{AB}]/p(a|s)$, where $\rho_{AB}$ denotes the entangled state shared between Alice and Bob and $p(a|s)=\textrm{Tr}_{AB}[(E_a^{(s)}\otimes \mathbb{I})\rho_{AB}]$. Because of no-signalling, Bob can never infer the measurement setting of Alice $s$, and therefore it holds that the different ensemble preparations labelled by $s$ will give the same state $\rho_B=\sum_a p(a|s) \rho_{a|s}$ for every $s$. These are the operational equivalences that we consider.
Let us now apply preparation noncontextuality to the operationally equivalent preparations of $\rho_B.$ It reads as 
\begin{equation*}\sum_a p(a|s)p(\lambda|s,a)=p(\lambda) \;\; \forall s.\end{equation*}
From the basic law for writing joint probabilities in terms of conditional probabilities we notice that $p(a|s)p(\lambda|s,a)=p(a,\lambda|s)$ and, in turn, $p(a,\lambda|s)=p(a|\lambda,s)p(\lambda|s).$ We can then write \begin{align*}\sum_a p(a|s)p(\lambda|s,a) & =\sum_a p(a|\lambda,s)p(\lambda|s)\\&=p(\lambda|s) \sum_a p(a|\lambda,s)=p(\lambda|s),
\end{align*}
and so \begin{equation} \label{PNCproof}p(\lambda|s)=p(\lambda). \end{equation}
 
At this point, let us consider the joint conditional probability in Bell scenario $p(a,b|s,t).$ We want to show that it can be written as $p(a,b|s,t)=\sum_{\lambda}p(a|s,\lambda)p(b|t,\lambda)p(\lambda)$, as this defines the set of local correlations. 
First, we use the basic law for writing joint probabilities as conditional probabilities, \begin{equation*}p(a,b|s,t)=p(a|s,t)p(b|s,t,a).\end{equation*} We recall that no-signalling holds, and so $p(a|s,t)=p(a|s)$, thus yielding $p(a,b|s,t)=p(a|s)p(b|s,t,a)$.

Let us focus on $p(b|s,t,a)$ and write it introducing $\lambda$ as \begin{align*}p(b|s,t,a)&=\sum_{\lambda}p(b,\lambda|s,t,a)\\&=\sum_{\lambda}p(b,|\lambda,s,t,a)p(\lambda|s,t,a).\end{align*} 
We now impose two assumptions characterizing the ontological model framework \cite{CataniLeifer2020}. The first is known as measurement independence (usually justified as an assumption of no-retrocausality)  and, in the prepare and measure scenario on Bob's wing where the measurement settings are denoted with $t$, it reads as $p(\lambda|s,a,t)=p(\lambda|s,a).$ 
The second is known as $\lambda-$mediation (\textit{i.e.}, the ontic state $\lambda$ mediates any correlation between preparation and measurement) and, in the prepare and measure scenario on Bob's wing where the preparations are associated with $(s,a)$ and the measurement settings with $t$, it reads as $p(b,|\lambda,s,t,a)=p(b|\lambda,t).$ 
Therefore we obtain \begin{equation*}p(b|s,t,a)=\sum_{\lambda}p(b,|\lambda,t)p(\lambda|s,a),\end{equation*}
and so
\begin{equation*}p(a,b|s,t)=\sum_{\lambda}p(a|s)p(b,|\lambda,t)p(\lambda|s,a).\end{equation*}
By exploiting again the fact proven above that $p(a|s)p(\lambda|s,a)=p(a|\lambda,s)p(\lambda|s)$ we have that 
\begin{equation*}p(a,b|s,t)=\sum_{\lambda}p(b,|\lambda,t)p(a|\lambda,s)p(\lambda|s),\end{equation*} and, by using preparation noncontextuality as in Eq.~\eqref{PNCproof}, we end the proof, \begin{equation*}p(a,b|s,t)=\sum_{\lambda}p(a|s,\lambda)p(b|t,\lambda)p(\lambda).\end{equation*}

Let us conclude by stressing how the proof just provided does not connect an XOR game to the corresponding dual XOR* game, as the inputs considered in the preparation stage of the prepare and measure scenario associated with the XOR* game are $(s,a)$ and not just $s$. As a consequence, this proof does not apply to the CHSH* game, that indeed cannot involve a proof of preparation contextuality due to its too low cardinality of the inputs.

\end{document}